\documentclass[10pt,leqno,twocolumn,nofootinbib]{revtex4-2}
\usepackage{geometry}
\newgeometry{vmargin={25mm}, hmargin={25mm,25mm}}   

\usepackage{xcolor}

\usepackage[hidelinks]{hyperref}

\usepackage{amsmath,amsthm,amsfonts,amssymb,latexsym,amscd,enumerate}
\usepackage{slashed}

\usepackage{todonotes}

\usepackage{palatino}
\usepackage[mathcal]{euler}

 \usepackage{natbib}

 \usepackage{tikz}
\usetikzlibrary{arrows.meta}
\usetikzlibrary{cd}

\usepackage{tikz-cd}
\usetikzlibrary{decorations.markings}
\tikzset{->-/.style={decoration={
  markings,
  mark=at position #1 with {\arrow{>}}},postaction={decorate}}}
\tikzset{->-/.default=0.5}

\usepackage{pgfplots}

\pgfplotsset{compat=1.10}
\usepgfplotslibrary{fillbetween}
\usetikzlibrary{patterns}

\numberwithin{equation}{section}
\swapnumbers

\newtheorem{theorem}{Theorem}[section]
\newtheorem*{theorem*}{Theorem}

\newtheorem*{proposition*}{Proposition}
\newtheorem{lemma}[theorem]{Lemma}

\theoremstyle{definition}
\newtheorem{definition}[theorem]{Definition}
\newtheorem{remark}[theorem]{Remark}
\newtheorem{interpretation}[theorem]{Physical interpretation}

\newtheorem{example}[theorem]{Example}

\newcommand{\Tr}{\mathrm{Tr}}

\newcommand{\BB}{\mathrm{B}}

\newcommand{\spec}{\mathrm{spec}}

\newcommand{\Roe}{\mathrm{Roe}}
\newcommand{\C}{\mathbb{C}}
\newcommand{\R}{\mathbb{R}}

\newcommand{\Z}{\mathbb{Z}}

\DeclareMathOperator{\Exp}{\mathrm{Exp}}

\begin{document}

\title{Coarse geometry and its applications in solid state physics}
\collaboration{Entry for the {\itshape{Encyclopedia of Mathematical Physics}} 2nd edition}
\author{Matthias Ludewig}
\affiliation{Universit\"at Regensburg}

\maketitle

Topological insulators (or topological phases), are certain materials or meta-materials which have a certain topological non-triviality in their mathematical description.
For a general introduction, we refer to the broader physics literature and the article by Bernevig in this encyclopedia.
The use of K-theory in the classification of topological phases goes back to Kitaev \cite{Kitaev} and is now fairly standard in the mathematical physics literature \cite{FreedMoore,ThiangKtheory}.
Here K-theory typically refers to topological K-theory of the Bloch bundle on the Brillouin zone; see \S\ref{SectionBlochTheory}.
However, as observed by Bellissard around 40 years ago \cite{Belissard,BelissardVanElstSchuba}, non-commutative $C^*$-algebras of observables are needed in order to classify disordered systems, which naturally leads to the use of operator algebra K-theory.\footnote{The $C^*$-algebras suggested by Belissard are certain crossed product algebras (also used in the modern treatment \cite{ProdanSchulzBaldes}).}
It has been realized fairly recently that \emph{Roe algebras}, which come from the mathematical subject of \emph{coarse geometry}, are a particularly good choice here, as they are physically well-motivated and on the other hand allow simple proofs of many features of topological insulators \cite{KubotaControlled, EwertMeyer,LudewigThiangGaplessness,KubotaLudewigThiang,LudewigThiangLocalization}. 
In this article, we give an overview over recent developments in this direction.

One physical feature of many topological insulators is that they are insulating in their interior (the \emph{bulk}), but extremely well-conducting on the boundary.
Mathematically, this corresponds to filling of the bulk Hamiltonian's spectral gap when introducing boundary (the new spectrum corresponding to certain boundary-localized states \cite[\S2.4.3]{ProdanSchulzBaldes}).
Moreover, in two dimensional topological phases, the mathematical models predict certain quantized edge currents \cite{Halperin,KellendonkSchuba,ThiangEdgeFollowing}, which have also been observed in experiments \cite{Exciton-polariton,PNAS,MNHTI}.
These phenomenona are robust against all kinds of perturbations.
In this article, we will focus on the above effects, which can be well described using coarse geometry.

In this article, we consider the case of no internal symmetries (i.e., topological phases of type A in the Cartan classification, see \cite{AltlandZirnbauer}), which corresponds to the use of the usual complex Roe algebras.
To encorporate time reversal, chiral and/or particle-hole symmetry, one uses real versions of Roe algebras and their KO-theory, as discussed in \cite{EwertMeyer}.
Many of the methods discussed below carry over to this setup without changes, but some aspects of the theory need a different treatment due to the presence of torsion in the KO-groups.
These more refined discussions are beyond the scope of this expository article.

\section{General mathematical setup}
\label{SectionSetup}

In this section, we introduce the general non-commutative geometry setup for the description of topological phases used in this article.
The choice of a coarse geometric algebra of observables will be discussed later, in \S\ref{SectionRoeAlgebras}. 

Consider a Hilbert space $\mathscr{H}$ with a \emph{Hamiltonian} $H$, be which we a self-adjoint (possibly unbounded) operator on $\mathscr{H}$.
For simplicity, we also assume also that $H$ is positive.

\begin{definition}
\label{DefinitionInsulator}
We say that $H$ is an \emph{insulator at energy level} $E \in \R$, if $E$ is not contained in the spectrum $\spec(H)$ of $H$.
We write $p_E$ for the orthogonal projection of $H$ onto the spectral subspace of $\mathscr{H}$ corresponding to the spectrum of $H$ below $E$.
\end{definition}

In physics, this value $E$ is typically called the \emph{Fermi energy} of the physical system whose dynamics are described by $H$ and we will refer to $p_E$ as the \emph{Fermi projection}.

Let now $\mathcal{A} \subset \BB(\mathscr{H})$ be a $C^*$-algebra of bounded operators on $\mathscr{H}$, which we interpret as the  algebra of observables of the system.
We require that $H$ is compatible with $\mathcal{A}$ in the sense that for each compactly supported continuous function $f \in C_c(\R)$, the operator $f(H)$ (formed using functional calculus) is contained in $\mathcal{A}$. 
Observe that if $H$ is bounded, this is equivalent to requiring $H \in \mathcal{A}$.
In general, this assumption ensures that if $H$ is insulating at energy level $E$, then the spectral projection $p_E$ is contained in $\mathcal{A}$; indeed, since $H$ has a spectral gap at $E$ and is positive, we can write $p_E = f(H)$ for $f \in C_c(\R)$.\footnote{While we always have $p_E = f(H)$ for $f$ the (compactly supported) indicator function of $[0, E]$, the point here is that the spectral gap allows to choose $f$ to be continuous. This is important as the observable algebra $\mathcal{A}$ typically does not contain the operators $f(H)$ for non-continuous $f$.}

\begin{definition}
\label{DefinitionTopologicalInsulator}
We say that $H$ is a \emph{topological insulator} if the corresponding class $[p_E] \in K_0(\mathcal{A})$ in $C^*$-algebra K-theory is non-trivial.
\end{definition}

We recall here that the K-theory group $K_0(\mathcal{A})$ consists of formal differences of homotopy classes of projections in matrix algebras over the unitalization of $\mathcal{A}$; see, e.g., \cite{WeggeOlsen} or \cite{Blackadar} for an introduction to the K-theory of operator algebras.

\begin{remark}
More precisely, Definition~\ref{DefinitionTopologicalInsulator} is actually that of a \emph{type A topological insulator} (i.e., a topological insulator with \emph{no} internal symmetries) and has to be modified in the presence of internal symmetries.
In particular, in the presence of time reversal symmetry, the $C^*$-algebra $\mathscr{A}$ will be real and one has to use KO-theory instead of K-theory.
While these groups can be generally described by unitaries and projections with certain additional symmetries \cite{BoersemaLoring}, it may then be convenient to use other descriptions of K-theory, e.g.\ that of van Daele, as advocated for in \cite{JosephMeyer}.
\end{remark}

\begin{example}
In applications, the Hilbert space $\mathscr{H}$ is typically either the space $L^2(\R^d)$ of square-integrable functions on $\R^d$ or -- in the so-called \emph{tight-binding approximation} -- its discretization $\ell^2(\Z^d) \otimes \C^n$.
In the first case, typical examples of Hamiltonians are magnetic Schr\"odinger operators, 
\begin{equation}
\label{MagneticSchroedinger}
H = (d - iA)^*(d-iA),
\end{equation}
where $A \in \Omega^1(\R^d)$ is a magnetic potential.
In the second case, the Hamiltonian $H$ is typically bounded, given by an infinite matrix $(H_{ij})_{i,j \in \Z^d}$, which typically has the \emph{finite propagation} property that there exists $r >0$ such that $H_{ij} = 0$ whenever $|i - j| \geq r$.
In other words, only near sites interact.
Schr\"odinger operators with bounded potential can also be considered; see, \cite[Prop.~2.5]{EwertMeyer}.
\end{example}

While the Hilbert space and the Hamiltonian are typically prescribed by the experiment, the algebra $\mathcal{A}$ is a choice one has to make.
This choice should be made according to the properties of the Hamiltonian and depends on the physical phenoma one seeks to describe.
In the coarse geometric approach to topological insulators, $\mathcal{A}$ is typically taken to be the \emph{Roe algebra} of the underlying space; see \S\ref{SectionRoeAlgebras}.

\section{Connection to Bloch theory}
\label{SectionBlochTheory}

In this section, we compare the approach from \S\ref{SectionSetup} to the classical approach to topological insulators using Bloch theory.
Here the Hamiltonian $H$ is assumed to be periodic, in the sense that it commutes with lattice translations.

Let us for simplicity use a the discrete model, where $\mathscr{H} = \ell^2(\Z^d)\otimes \C^n$.
Then by Fourier transform, the Hamiltonian $H$ corresponds to a smooth matrix-valued function $\hat{H} : T^d \to M_n(\C)$ on the torus $T^d$ (the so-called \emph{Brillouin zone}).
The spectrum of $H$ is just the union of all $\spec(\hat{H}(\mathbf{k}))$, $\mathbf{k} \in T^d$, hence a spectral gap of $H$ at $E \in \R$ implies that none of the matrices $H(\mathbf{k})$ has eigenvalue $E$.
Therefore the vector spaces
\begin{equation*}
  V^E_{\mathbf{k}} = \bigoplus_{\lambda \leq E} \mathrm{Eig}(H(\mathbf{k}), \lambda), \qquad \mathbf{k} \in T^d,
\end{equation*}
glue together to a smooth vector bundle $V^E$ over $T^d$ (the \emph{Bloch bundle} of $H$), and topological insulators are those having non-trivial Bloch bundle.
While this non-triviality can be measured using topological invariants such as characteristic numbers, the K-theory class $[V^E] \in K^0(T^d)$ should be considered as the more fundamental invariant.\footnote{Recall that two vector bundles $V$ and $W$ define the same class in $K_0(X)$ if they become isomorphic after taking the direct sum with a trivial vector bundle on both sides. 
This stabilization corresponds in physics to adding several trivial ``bands''.}

These classical considerations correspond to taking $\mathcal{A} = C^*_r(\Z^d) \otimes M_n(\C)$ in the mathematical setup of \S\ref{SectionSetup}, where $C^*_r(\Z^d)$ is the (reduced) group $C^*$-algebra of $\Z^d$ and $\mathcal{A}$ acts on $\mathscr{H} = \ell^2(\Z^d) \otimes \C^n$ via the left regular representation and matrix multiplication.
Namely, it is well known that the topological K-theory group $K^0(T^d)$ is isomorphic to the $C^*$-algebra K-theory group of the algebra $C(T^d)$ of continuous functions on the torus. 
This $C^*$-algebra in turn is (via Fourier transform), isomorphic to $C^*_r(\Z^d)$.
The corresponding isomorphism 
\begin{equation*}
K^0(T^d) \cong K_0(C(T^d)) \cong  K_0(C^*_r(\Z^d))
\end{equation*}
 sends the class $[V^E]$ of the Bloch bundle to the class defined by the Fermi projection $p_E$.

As pointed out by Belissard \cite{Belissard}, the assumption of periodicity on $H$ is only an idealization, which can never be satisfied in physical practice. 
So even if the model Hamiltonian of interest is periodic, it is crucial for $\mathcal{A}$ to contain sufficiently small non-periodic perturbations of $H$.
This persistence is usually taken for granted in the physics literature. 
In reality, only certain "strong" invariants from periodic models can survive non-periodic perturbations, and we shall explain how Roe algebras provide the mathematical basis to understand this.

\section{Roe algebras}
\label{SectionRoeAlgebras}

In this section, we introduce Roe algebras, which are the observable algebra of choice for this article.
General references for this section are \cite[\S6.3]{HigsonRoe}, \cite[\S5]{WillettYu} or \cite[\S3]{RoeIndexTheory}.

Let $(X, d)$ be a metric space.
By an \emph{$X$-module}, we mean a Hilbert space $\mathscr{H}$ together with an action of the $C^*$-algebra $C_0(X)$ of continuous functions on $X$ vanishing at infinity.
We say that $\mathscr{H}$ is \emph{ample} if $f \in C_0(X)$ acts as a compact operator only when $f =0$.

\begin{example}
In the case that $X$ is $\R^d$, or more generally, a Riemannian manifold, the reader should think of the $X$-module  $\mathscr{H} = L^2(\R^d)$, with the usual action of $C_0(X)$ by multiplication.
On the other hand, for $X$ a discrete metric space, the space $\mathscr{H} = \ell^2(\Z^d)$ is \emph{not} ample.
However, one may tensor with some auxiliary ``coefficient'' Hilbert space to obtain an ample $X$-module.
In physics, this amounts to increasing internal degrees of freedom.
\end{example}

By the spectral theorem for commutative $C^*$-algebras \cite[Thm.~12.22]{Rudin}, the action of $C_0(X)$ extends to an action of the algebra of bounded Borel functions on $\mathscr{H}$ (see also \cite[Proposition 1.6.11]{WillettYu}).
In particular, for any Borel subset of $W \subset X$ the indicator function $\chi_W$ acts as a projection operator on $\mathscr{H}$.
Throughout, for $f$ a bounded Borel function on $X$, we also denote by $f$ the operator on $\mathscr{H}$ provided by the action.

An operator $T \in \BB(\mathscr{H})$ is \emph{locally compact}, if $\chi_W T$ and $T\chi_W$ are compact for all bounded subsets $W \subseteq X$. 
$T$ is called of \emph{finite propagation} if there exists $r\geq 0$ such that $\chi_V T \chi_W = 0$ whenever $d(V, W) > r$.

\begin{definition}
Let $\mathscr{H}$ be an ample $X$-module.
The \emph{Roe algebra} $C^*_{\Roe}(X)$ is the norm-closure in $\BB(\mathscr{H})$ of the subalgebra of all locally compact, finite propagation operators on $\mathscr{H}$.
\end{definition}

We do not mention the $X$-module $\mathscr{H}$ in notation as the isomorphism type of $C^*_{\Roe}(X)$ is independent of the choice of the ample $X$-module $\mathcal{H}$.
Precisely, if $\mathcal{H}_1$ and $\mathcal{H}_2$ are two different ample $X$-modules, then the isomorphism of Roe algebras is implemented by a unitary transformation $U : \mathcal{H}_1 \to \mathcal{H}_2$; see \cite[Prop.~6.3.12]{HigsonRoe} or \cite[Theorem 2.1]{EwertMeyer}.

\begin{example}
The Roe algebra depends only on the large scale structure of the space $X$.
In particular, for any bounded space $X$, the finite propagation property is irrelevant, hence in this case, $C^*_{\Roe}(X) = K(\mathscr{H})$, the algebra of compact operators on $\mathscr{H}$.
\end{example}

A subset $Y \subset X$ is called \emph{coarsely dense} if there exists $r>0$ such that $B_r(Y) = X$; here $B_r(Y)$ denotes the  set of $x \in X$ such that $d(x, y) \leq r$ for some $y \in Y$.
The Roe algebra of $X$ is isomorphic to the Roe algebra of any coarsely dense subset \cite[Thm.~2.7]{EwertMeyer}, where the isomorphism is again implemented by a unitary isomorphism of the underlying Hilbert spaces.
For example the inclusion $\Z^d \hookrightarrow \R^d$ is a coarse equivalence, which allows easy comparison of discretized and continuum models.
In particular, we have
\begin{equation}
\label{KtheoryRoeAlgebraRn}
  K_0(C^*_{\Roe}(X)) = \begin{cases} \Z & \text{if $d$ is even} \\ 0 & \text{if $d$ is odd}
  \end{cases}
\end{equation}
for $X$ either $\Z^d$ or $\R^d$.
More generally, the Roe algebra of two \emph{coarsely equivalent} spaces is isomorphic; see Definition 2.8 and Thm.~2.9 of \cite{EwertMeyer}.


\begin{remark}
One can define the Roe algebra for the more general class of \emph{bornological coarse spaces}.
We refer to the entry on ``Coarse geometry'' in this encyclopedia for a detailed account; see also \cite{BunkeEngelHomotopy}.
\end{remark}


\medskip

We now discuss localized Roe algebras, which are subalgebras of the Roe algebras supported near a subset.
Here it is convenient to use the following notion, which was introduced in \cite[Definition 3.2]{BunkeEngelHomotopy}.

\begin{definition}
A \emph{big family} in $X$ is a collection $\mathcal{Y}$ of subsets of $X$ such that for $Y_1, Y_2 \in \mathcal{Y}$ and $r>0$, there exists $Y^\prime \in \mathcal{Y}$ with $B_r(Y_1) \cup B_r(Y_2) \subseteq Y^\prime$.
\end{definition}

In particular, for any subset $Y \subset X$, we have the big family 
\begin{equation}
\label{SingletonBigFamily}
\{Y\} := \{Z \subset X \mid \exists r>0 : Z \subseteq B_r(Y)\},
\end{equation}
comprising all uniform thickenings of $Y$.
Given a big family $\mathcal{Y}$ in $X$, we say that $T \in C^*_{\Roe}(X)$ is \emph{supported in} $\mathcal{Y}$, if there exists $Y \in \mathcal{Y}$ such that $fT = Tf = 0$ for all $f$ with support in $X \setminus Y$.

\begin{definition}
Given a big family $\mathcal{Y}$, the \emph{Roe algebra localized at $\mathcal{Y}$}, denoted by $C^*_{\Roe}(\mathcal{Y})$, is the subalgebra of $C^*_{\Roe}(X)$ obtained by taking the closure of all operators supported in $\mathcal{Y}$.
\end{definition}

For any big family $\mathcal{Y}$, the localized Roe algebra $C^*_{\Roe}(\mathcal{Y})$ is a two-sided ideal in $C^*_{\Roe}(X)$.
For subsets $Y \subseteq X$, the K-theory of $C^*_{\Roe}(\{Y\})$ coincides with that of $C^*_{\Roe}(Y)$ (taken without reference to the ambient space $X$); see the proof of Thm.~9.2 in \cite{RoeIndexTheory}.
It is straightforward to see that the intersection 
\begin{equation*}
\mathcal{Y} \cap \mathcal{Y}^\prime := \{Y \cap Y^\prime \mid Y \in \mathcal{Y}, Y^\prime \in \mathcal{Y}^\prime\}
\end{equation*}
 of two big families $\mathcal{Y}$ and $\mathcal{Y}^\prime$ is always a big family and we have
\begin{equation}
\label{IntersectionOfFamilyAlgebras}
C^*_{\Roe}(\mathcal{Y}) \cap C^*_{\Roe}(\mathcal{Y}^\prime) = C^*_{\Roe}(\mathcal{Y}\cap \mathcal{Y}^\prime).
\end{equation}

\medskip

Finally, we discuss equivariant Roe algebras; see \cite[\S5.2]{WillettYu} for a general reference.
Suppose that a countable discrete group $\Gamma$ acts properly on a metric space $X$ by isometries and let $\mathscr{H}$ be an $X$-module with a compatible unitary representation $U: \Gamma \to \mathrm{U}(\mathscr{H})$, meaning that $\gamma^*f = U_{\gamma^{-1}} f U_\gamma$ for all $f \in C_0(X)$ and all $\gamma \in \Gamma$. 

\begin{definition}
The \emph{equivariant Roe algebra} $C^*_{\Roe}(X)^\Gamma$, is the norm closure of the algebra of all $\Gamma$-equivariant locally compact operators of finite propagation on $\mathscr{H}$.
\end{definition}

If the group action is such that the quotient $X / \Gamma$ is compact, then the choice of any bounded fundamental domain $F \subset X$ for the $\Gamma$-action yields an isomorphism
\begin{equation*}
C^*_{\Roe}(X) \cong C^*_r(\Gamma) \otimes \mathrm{K}(\mathscr{H}|_{F}),
\end{equation*}
where $\mathscr{H}|_{F} = \chi_F\mathscr{H}$ denotes the restriction of the $X$-module to $F$ and $\mathrm{K}(\mathscr{H}|_{F})$ denotes the corresponding algebra of compact operators \cite[\S5.2]{WillettYu}.
As tensoring with the latter does not change K-theory, the inclusion $C^*_{\Roe}(X)^\Gamma \hookrightarrow C^*_{\Roe}(X)$ induces a map
\begin{equation}
\label{MapToRoeAlgebra}
  K_0(C^*_r(\Gamma)) \longrightarrow K_0(C^*_{\Roe}(X)).
\end{equation}
In the example $X = \R^d$ or $\Z^d$ and $\Gamma = \Z^d$ with $d = 2m$ even, this results in a map
\begin{equation}
\label{KTheoryMapToRoeAlgebra}
   \underbrace{K^0(T^d)}_{\Z^{2^{d-1}}} \cong K_0(C^*_r(\Z^d)) \longrightarrow \underbrace{K_0(C^*_{\Roe}(X))}_\Z
\end{equation}
which sends the complex vector bundle with $m$-th Chern class equal to the generator of $H^d(T^d, \Z) \cong \Z$ to the generator of the right hand side.

As a non-trivial class $[p_E]$ on the left of \eqref{KTheoryMapToRoeAlgebra} may become trivial on the right, we observe that not all periodically-topological phenomena are tractable using Roe algebras.
However, when $[p_E]$ defines a non-trivial element in the K-theory of the Roe algebra (which happens, e.g., for \emph{strong} topological insulators, in physics terminology), the coarse geometric approach is viable.

\section{Example Hamiltonians}
\label{SectionExamples}

In this section, we discuss examples of topological insulators (see Definition~\ref{DefinitionTopologicalInsulator}) where the observable algebra $\mathscr{A}$ is the Roe algebra of the underlying space.
In view of \eqref{KTheoryMapToRoeAlgebra} many non-trivial examples arise from periodic Hamiltonian on lattice systems.
Concretely, let $H = (H_{ij})_{i, j \in \Z}$, $H_{ij} \in M_n(\C)$, be a lattice-periodic Hamiltonian on $\ell^2(\Z^d) \otimes \C^n$; in other words, $H \in C^*_r(\Z^d) \otimes M_n(\C)$.

The 2-dimensional case is of particular interest, relating to the quantum Hall effect; particular model Hamiltonian have been given by Haldane \cite[Eq.\ (1)]{Haldane}; see also Kane and Mele \cite[Eq.\ (6)]{KaneMele}.
A model Hamiltonian for a 2-dimensional gyroscopic metamaterial has been introduced in \cite[Eq.\ (2)]{PNAS}.

\begin{example}
A simple toy Hamiltonian that works in any even dimension was suggested in \cite[Eq.\ (2.24)]{ProdanSchulzBaldes}.
It acts on $\ell^2(\Z^d) \otimes \C^{d/2}$ and is given by
\begin{small}
\begin{equation*}
H = \frac{1}{2i} \sum_{j=1}^d (S_j - S_j^*) \otimes \gamma_j + \Biggl( m + \sum_{j=1}^d (S_j + S_j^*)\Biggr) \otimes \gamma_0,
\end{equation*}
\end{small}where $S_j$ denotes the shift in the $j$-th coordinate direction and $\gamma_1, \dots, \gamma_d$ are Clifford algebra generators and $\gamma_0 = i^{d/2}\gamma_1 \cdots \gamma_d$.
If $m \notin \{-n, -n+2, \dots, n-2, n\}$, this Hamiltonian has a spectral gap at zero and hence defines a class $[p_H] \in K_0(C^*_r(\Z^d))$ and, by \eqref{MapToRoeAlgebra}, in $C^*_{\Roe}(\Z^d)$.
One can show that this class is non-trivial if $|m| < n$ \cite{JosephMeyer}.
\end{example}

We now give examples for topological insulators in continuum models, where the Hilbert space is $\mathscr{H} = L^2(X)$ for a 2-dimensional Riemannian manifold $X$, with observable algebra $\mathcal{A} = C^*_{\Roe}(X)$.
An interesting class of examples is the Landau Hamiltonian $H_{\mathrm{Lan}}$.
This is by definition the magnetic Schr\"odinger operator \eqref{MagneticSchroedinger} for a magnetic potential $A$ such that 
$dA = b \cdot \mathrm{vol}_X$ for some $b \in \R \setminus \{0\}$, where $\mathrm{vol}_X$ is the Riemannian volume form of $X$.

\begin{example}
\label{ExampleEuclideanLandau}
In the case of $X = \R^2$ with the flat metric, one can show that the spectrum of the Landau Hamiltonian $H_{\mathrm{Lan}}$ is quantized, given by
\begin{equation}
\label{LandauSpectrum}
\mathrm{spec}(H_{\mathrm{Lan}}) = \{(2n-1)|b| \mid n=1, 2, 3, \dots\}.
\end{equation}
Here each of the eigenvalues (called \emph{Landau levels}) is infinitely degenerate. 
One can show that the spectral projection onto each eigenspace of $H_{\mathrm{Lan}}$ is a generator of $K_0(C^*_{\Roe}(\R^2))$ {\normalfont \cite[\S2.3 \& Thm.~3]{LudewigThiangGaplessness}}.
\end{example}

\begin{example}
\label{ExampleHyperbolidLandau}
For $X = \mathbb{H}^2$, the hyperbolic space, the Landau Hamiltonian $H_{\mathrm{Lan}}$ has the finite set of eigenvalues
\begin{equation*}
  (2n-1)|b|-n(n-1), \quad n=1, \dots, n_{\max} < |b| + \frac{1}{2},
\end{equation*}
each of infinite multiplicity, as well as continuous spectrum $[\frac{1}{4} + b^2, \infty)$, see \cite{ComtetHouston}.
One can show that similar to \eqref{KtheoryRoeAlgebraRn}, we have $K_0(C^*_{\Roe}(\mathbb{H}^2)) \cong \Z$ and that the spectral projection onto each eigenspace of $H_{\mathrm{Lan}}$ is a generator \cite[Lemma~3.1 \& Thm.~3]{LudewigThiangGaplessness}.
%
\end{example}

\begin{example}
Example~\ref{ExampleEuclideanLandau} can be generalized to the situation where $X$ is a complete Riemannian surface with not necessarily constant scalar curvature.
We may consider a magnetic Hamiltonian \eqref{MagneticSchroedinger} with magnetic potential $A$ satisfying $dA = \theta \cdot \mathrm{vol}_X$ for some smooth function $\theta$ such that $\theta - b$ vanishes at infinity, where $b \neq 0$ is a constant.
Then under suitable assumptions on $X$ (non-vanishing of the \emph{coarse index of the Dirac operator}; see the article ``Coarse geometry''), the essential spectrum of $H$ is again the the discrete set \eqref{LandauSpectrum}, and for each essential eigenvalue, the corresponding eigenprojection is a generator for $K_0(C^*_{\Roe}(X))$ \cite[Prop.~3.20]{KubotaBulkDislocation}.
\end{example}

\medskip

\section{Bulk-Boundary correspondence}

So far, we only discussed the ``bulk'' of the material, i.e., its idealized, boundaryless version.
We now turn to investigating the behavior of a topological insulator at the boundary.

Let $X$ be a metric space (which we view as the ``bulk'' of our material), and let $\mathscr{H}$ be an ample $X$-module.
Let moreover $H$ be a Hamiltonian on $\mathscr{H}$ satisfying the assumptions of \S\ref{SectionSetup} with respect to $\mathcal{A} = C^*_{\Roe}(X)$, i.e., we assume that $f(H)$ is contained in the Roe algebra of $X$ for each $f \in C_c(\R)$.

For a subset $Y \subset X$, denote by 
\begin{equation}
\label{CoarseBoundary}
\partial Y := \{Y\} \cap \{X \setminus Y\}
\end{equation}
 the \emph{coarse boundary} of $Y$, a big family on $X$.
We then have an exact 6-term sequence of K-theory groups
\begin{widetext}
\begin{equation}
\label{Mayer-Vietoris}
\begin{tikzcd}
  K_0(C^*_{\Roe}(\partial Y)) \ar[r] & K_0(C^*_{\Roe}(\{Y\})) \oplus K_0(C^*_{\Roe}(\{X \setminus Y\})) \ar[r] & K_0(C^*_{\Roe}(X))\ar[d, "\partial"] \\
  K_1(C^*_{\Roe}(X)) \ar[u, "\partial"] & K_1(C^*_{\Roe}(\{Y\})) \oplus K_1(C^*_{\Roe}(\{X\setminus Y\})) \ar[l] & K_1(C^*_{\Roe}(\partial Y)) \ar[l]
\end{tikzcd}
\end{equation}
\end{widetext}
called the \emph{coarse Meyer-Vietoris sequence} \cite{HigsonRoeYu}.

\begin{remark}
The statement here is slightly more general then the statement obtained in \S5 of \cite{HigsonRoeYu}, which requires the notion of \emph{coarse transversality}, see \cite[\S3.4]{EwertMeyer}.
The result in the form stated above is essentially Prop.~8.82 of \cite{BunkeEngelHomotopy}.
It follows from Lemma~1 of \cite{HigsonRoeYu} (see also Prop.~3.6 of \cite{EwertMeyer}), by observing that 
\begin{equation*}
C^*_{\Roe}(\{Y\}) \cap C^*_{\Roe}(\{X \setminus Y\}) = C^*_{\Roe}(\partial Y), 
\end{equation*}
by \eqref{IntersectionOfFamilyAlgebras}.
\end{remark}

\begin{definition}
We say that for a subset $Y \subseteq X$, the \emph{bulk-boundary correspondence holds} if the Mayer-Vietoris boundary map 
\begin{equation}
\label{MayerVietorisBoundaryMap}
\partial : K_0(C^*_{\Roe}(X)) \longrightarrow K_1(C^*_{\Roe}(\partial Y))
\end{equation}
is an isomorphism.
\end{definition}

Validity of the bulk-boundary correspondence implies that non-triviality of a K-theory class determined by the Fermi projection $p_E$ of a Hamiltonian $H$ that is insulating at $E$ can be detected at the (coarse) boundary.
Actually, this is already true if the map \eqref{MayerVietorisBoundaryMap} injective.

\begin{remark}
  In the physical literature, a topological insulator is often \emph{defined} as a phase with particular properties of the boundary. 
 From the mathematical point of view, this is not necessarily equivalent to the definition via non-triviality in the bulk: The topological non-triviality may be detected at the boundary if the bulk-boundary correspondence holds, but this may fail in general.
   \end{remark}

\begin{theorem}
\label{TheoremFlasqueBulkBoundary}
  If $Y$ and $X \setminus Y$ are both \emph{flasque}, then bulk-boundary correspondence holds.
\end{theorem}

See \cite[Def.~9.3]{RoeIndexTheory} for the general definition of the flasqueness property.
In particular, any half space in $\R^d$ is flasque, as well as any subspace coarsely equivalent to a half space. 
Hence this result applies to a wide class of examples.

\begin{proof}[Proof of Thm.~\ref{TheoremFlasqueBulkBoundary}]
The Roe algebra K-theory of flasque spaces vanishes \cite[Prop.~9.4]{RoeIndexTheory}, so the result follows from exactness of \eqref{Mayer-Vietoris}.
\end{proof}

\section{Gap-filling}

In this section, we discuss a coarse geometric \emph{gap-filling result}, which says that for a topological insulator, the spectral gap of the bulk Hamiltonian closes upon introducing boundary.

Let $H$ be a Hamiltonian acting on an ample $X$-module $\mathscr{H}$ which is compatible with the Roe algebra $C^*_{\Roe}(X)$ defined using $\mathscr{H}$.

\begin{definition}
\label{DefinitionAdaptedOperator}
Let $Y \subset X$ be a subspace.
A \emph{Hamiltonian on $Y$ adapted to} $H$ is a self-adjoint (possibly unbounded) operator on $\mathscr{H}$ such that for all $f \in C_0(\R)$, we have $f(\tilde{H}) \in C^*_{\Roe}(\{Y\})$ and $f(H) - f(\tilde{H}) \in C^*_{\Roe}(\{X \setminus Y\})$.
\end{definition}

\begin{example}
\label{ExampleContraction}
If $H$ is bounded, so that $H \in C^*_{\Roe}(X)$, then the contraction $\tilde{H} = \chi_Y H \chi_Y$ of $H$ to $Y$ satisfies these requirements. 
\end{example}

\begin{example}
If $X$ is a complete Riemannian manifold, $\mathscr{H} = L^2(X)$, $Y \subset X$ an open subset and $H$ a magnetic Schr\"odinger operator as in \eqref{MagneticSchroedinger}, we can take $\tilde{H}$ to be the restriction of $H$ to $Y$ with Dirichlet boundary conditions (i.e., the closure of the essentially self-adjoint operator on $L^2(Y)$ with domain $C^\infty_c(Y)$, extended by zero to $L^2(X)$); see \cite[\S1.4]{LudewigThiangGaplessness}.
More general boundary conditions can be treated as well; see ibid., Remark 1.8.
\end{example}

\begin{theorem}
\label{ThmGapFilling}
Let $Y\subseteq X$ and let $\tilde{H}$ be a Hamiltonian on $Y$ adapted to $H$.
Suppose that $H$ is insulating at some energy level $E$ and let $[p_E] \in K_0(C^*_{\Roe}(X))$ be the class of the corresponding fermi projection.
Then if 
\begin{equation*}
\partial([p_E]) \neq 0 \in K_1(C^*_{\Roe}(\partial Y)), 
\end{equation*}
we have $E \in \spec(\tilde{H})$.
\end{theorem}

We recall that elements of the K-theory group $K_1(C^*_{\Roe}(\partial Y))$ are homotopy classes of unitaries in matrix algebras over the unitalization $C^*_{\Roe}(\partial Y)^+$.

\begin{proof}
Suppose that $E \notin \spec(\tilde{H})$, in other words $E$ is in the resolvent set of both $H$ and $\tilde{H}$.
Then since the resolvent set of an operator is open, we can write $p_E = f(H)$ for some $f \in C_c(\R)$ such that $f(\lambda) \in \{0, 1\}$ whenever $\lambda \in \mathrm{spec}(\tilde{H}) \cup \mathrm{spec}(H)$.
Then $\tilde{p}_E := f(\tilde{H}) \in C^*_{\Roe}(\{Y\})$ is also a projection.
\\
By the isomorphism
\begin{equation*}
C^*_{\Roe}(\{Y\}) / C^*_{\Roe}(\partial Y) \cong C^*_{\Roe}(X) / C^*_{\Roe}(\{X \setminus Y\}), 
\end{equation*}
we see that the quotient algebra in the short exact sequence
\begin{equation}
\label{SESpartialY}
  C^*_{\Roe}(\partial Y) \longrightarrow C^*_{\Roe}(\{Y\}) \stackrel{\pi}{\longrightarrow} C^*_{\Roe}(\{Y\}) / C^*_{\Roe}(\partial Y).
\end{equation}
receives a ``restriction'' homomorphism 
\begin{equation}
\label{Homomorphismr}
  r: C^*_{\Roe}(X) \longrightarrow C^*_{\Roe}(\{Y\}) / C^*_{\Roe}(\partial Y).
\end{equation}
As $\tilde{H}$ is adapted to $H$, we have $\pi(\tilde{p}_E) = r(p_E)$ in $C^*_{\Roe}(\{Y\}) / C^*_{\Roe}(\partial Y)$.
We consider the six-term sequence in K-theory of this short exact sequence.
Recall that generally, the boundary map
\begin{equation*}
\mathrm{Exp} : K_0(C^*_{\Roe}(\{Y\}) / C^*_{\Roe}(\partial Y)) \to K_1(C^*_{\Roe}(\partial Y))
\end{equation*}
of such a six-term sequence is given by
\begin{equation}
\label{FormulaForExponentialMap}
\mathrm{Exp}([q]) = [\exp(2 \pi i \tilde{q})],
\end{equation}
where $\tilde{q} \in C^*_{\Roe}(\{Y\})$ is a self-adjoint lift of the projection $q$ (see \cite[9.3.2]{Blackadar}).
In particular, $\mathrm{Exp}([q]) = 0$ if a lift $\tilde{q}$ can be found that is a projection.
We therefore conclude that $\mathrm{Exp}([r(p_E)]) = 0$, since $\tilde{p}_E$ is a projection that lifts $r(p_E) = \pi(\tilde{p}_E)$.
On the other hand, the Meyer-Vietoris boundary map is precisely the composition
\begin{equation*}
\begin{tikzcd}
K_0(C^*_{\Roe}(X)) \ar[r ,"\pi_*"] \ar[ddr, bend right=25, "\partial"'] \ar[dr, "r"]& K_0(C^*_{\Roe}(X) / C^*_{\Roe}(\{X \setminus Y\})) \ar[d, equal]\\
& K_0(C^*_{\Roe}(\{Y\}) / C^*_{\Roe}(\partial Y)) \ar[d, "\mathrm{Exp}"] \\
& K_1(C^*_{\Roe}(\partial Y)),
\end{tikzcd}
\end{equation*}
see the proof of \cite[Prop.~3.6]{EwertMeyer}.
We conclude that $\partial([p_E]) = \mathrm{Exp}([r(p_E)]) = 0$.
\end{proof}

\begin{example}
The K-theoretic exponential map has also been used to show gap-filling results in \cite{ThiangEdgeFollowing,LudewigThiangCobordism,KubotaLudewigThiang,LudewigThiangGaplessness}.
To the author's knowledge, this kind of argument was first used in \cite{SchubaKellendonkRichterSimultaneous}.
\end{example}

\begin{example}
Let $Y \subset \R^2$ be a half space (or any open subset coarsely equivalent to a half space) and consider $H = H_{\mathrm{Lan}}$, the Landau Hamiltonian; see \S\ref{SectionExamples}.
Then if $E$ lies between the $n$-th and the $(n+1)$-st Landau level, then $[p_E] \in K_0(C^*_{\Roe}(X)) \cong \Z$ is $n$ times a generator, and by Thm.~\ref{TheoremFlasqueBulkBoundary}, bulk boundary correspondence holds, so that $\partial([p_E]) \neq 0$.
We obtain that the Landau Hamiltonian on $Y$ (with, say Dirichlet boundary conditions) has no spectral gaps.
\end{example}

\begin{example}
Using Thm.~\ref{ThmGapFilling}, one can also show that the Landau Hamiltonian on coarse half spaces $Y$ of the hyperbolic plane $\mathbb{H}^2$ has no spectral gaps, even though such subspaces $Y$ are not flasque; see \cite{LudewigThiangGaplessness}.
\end{example}

Summarizing, it is a feature of the coarse geometric setup that one may obtain the above gap-filling results for a wide variety of subspaces $Y \subseteq X$, without having to delve into peculiarities of the involved $C^*$-algebras (as, e.g., in \cite{ThiangEdgeFollowing}).
Calculating the spectrum of $\tilde{H}$ explicitly is typically impossible for subspaces $Y$ other than the perfect half space.



\section{Edge-traveling}

In this section, we give a result on how to detect non-triviality of the class $\partial([p_E])$, as required in Thm.~\ref{ThmGapFilling} in two-dimensional topological insulators.
To this end, we want to construct a homomorphism $K_1(\partial Y) \to \Z$
that allows us to detect non-triviality of an abstract K-theory class.
As it turns out, this is closely related to edge-traveling phenomena.
The presentation roughly follows \cite{LudewigThiangCobordism}.

Let $X$ be a metric proper metric space and let $\mathscr{H}$ be an ample $X$-module. 

\begin{lemma}
\label{LemmaCommutator}
For any Borel subset $W \subseteq X$ and any $T \in C^*_{\Roe}(X)$, we have $[\chi_W, T] \in C^*_{\Roe}(\partial Y)$, where $\partial W$ is the coarse boundary of $W$, see \eqref{CoarseBoundary}.
\end{lemma}

\begin{proof}
If $T \in C^*_{\Roe}(X)$ is an operator of finite propagation, it is straightforward to verify that $[\chi_W, T]$ is supported in $\partial W$.
The statement for general $T \in C^*_{\Roe}(X)$ follows by continuity.
\end{proof}

Suppose from now on that we are given two Borel subsets $Y, W \subset X$ that are \emph{transversal} in the sense that $\partial Y \cap \partial W$ is bounded, meaning that any $Z \in \partial Y \cap \partial W$ is bounded.

\begin{example}
\label{ExampleLinePartition}
The prototypical example here is $X = \R^2$ or $\Z^2$ and
\begin{equation}
\label{StandardHalfSpaces}
 Y = \{(x, y) \mid x \geq 0\}, \quad W = \{(x, y) \mid y \geq 0\}.
 \end{equation}
 More generally, we can take $Y$ and $W$ to be any set such that $\partial Y = \{L_1\}$, $\partial W = \{L_2\}$ for non-parallel affine lines $L_1, L_2 \subset \R^2$.
 \begin{equation*}
 \underbrace
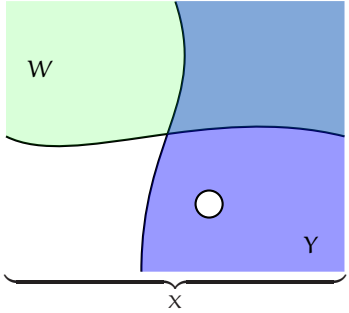
{
 \begin{tikzpicture}[scale=0.9, every node/.style={scale=0.9}]
\draw[thick, name path=A] (0,0) .. controls (0,2) and (1,2.5) .. (0.5,4);
\draw[white, name path =B] (3,0) -- (3,4);
\draw[thick, name path=C] (-2,2) .. controls (-1,1.5) and (1,2.5) .. (3,2); 
\draw[white, name path =D] (-2,4) -- (3,4);
\tikzfillbetween[of=A and B]{blue, opacity=0.4}
\tikzfillbetween[of=D and C]{green, opacity=0.15}
\filldraw[thick, white, draw=black] (1,1) circle (0.2cm);
\node at (-1.5,3) {$W$};
\node at (2.5,0.4) {$Y$};
 \end{tikzpicture}
 }_X
 \end{equation*}
 \end{example}

\begin{lemma}
For any unitary $u$ in $C^*_{\Roe}(\partial Y)^+ \otimes M_n(\C)$, the operator
\begin{equation*}
F_u := \chi_W u  \chi_W + \chi_{X \setminus W}
\end{equation*}
is a Fredholm operator on $\mathscr{H}$.
\end{lemma}

\begin{proof}
We have 
\begin{equation}
\label{FFinverse}
F_u F_{u^*} = 1 + \chi_W u[\chi_W, u^*]\chi_W.
\end{equation}
By Lemma~\ref{LemmaCommutator} (which clearly also holds for $T$ in the unitalization $C^*_{\Roe}(X)^+$, as well as matrix algebras over it), the commutator $[\chi_W, u^*]$ is contained in $C^*_{\Roe}(\partial Y \cap \partial W)$.
But since $\partial Y \cap \partial W$ is bounded, we have $C^*_{\Roe}(\partial Y \cap \partial W) = K(\mathscr{H})$ the algebra of compact operators on $\mathscr{H}$.
Equation \eqref{FFinverse} therefore shows that $F_u$ is invertible up to a compact perturbation, hence Fredholm.
\end{proof}

It is now easy to check that the map sending a unitary $u \in C^*_{\Roe}(\partial Y)^+ \otimes M_n(\C)$ to the index of the Fredholm operator $F_u$ depends only on the equivalence class of $u$ in $K_1(C^*_{\Roe}(\partial Y))$ and yields a well-defined group homomorphism
\begin{equation*}
\begin{aligned}
  \theta_W : K_1(C^*_{\Roe}(\partial Y)) \to \Z, \qquad
      [u]  \mapsto \mathrm{index}(F_u).
  \end{aligned}
\end{equation*}
In other words, $\theta_W$ gives a possibility to ``measure'' the non-triviality of abstract $K_1$-classes.
One can show that if one replaces $W$ by $W^\prime$ with the property that both
\begin{equation*}
  \partial Y \cap \{W\} \cap \{X \setminus W^\prime\} \quad \text{and} \quad \partial Y \cap \{W^\prime\} \cap \{X \setminus W\}
\end{equation*}
are bounded, then $\theta_W = \theta_{W^\prime}$.
This property of $\theta_W$ is called \emph{cobordism invariance} in \cite[\S4]{LudewigThiangCobordism}.

\begin{example}
For the choice of $Y$ and $W$ from Example~\ref{ExampleLinePartition}, one can check that the map $\theta_W$ is an isomorphism \cite[\S5.2]{LudewigThiangCobordism}.
In particular, this implies that
\begin{equation}
\label{IsoK1Z}
K_1(C^*_{\Roe}(\partial Y) \cong K_1(C^*_{\Roe}(\R)) \cong \Z
\end{equation}
which can be abstractly checked using the Mayer-Vietoris sequence \eqref{Mayer-Vietoris}, by partitioning $\R$ into two rays, whose Roe algebra K-theory vanishes by flasqueness.
Combining the isomorphism \eqref{IsoK1Z} with the fact that the Mayer-Vietoris boundary map is an isomorphism in this case (Thm.~\ref{TheoremFlasqueBulkBoundary}), we obtain the isomorphism \begin{equation}
\label{IsoRoe}
K_0\bigl(C^*_{\Roe}(\R^2)\bigr) \cong \Z
\end{equation} 
from \eqref{KtheoryRoeAlgebraRn}.
\end{example}

\begin{remark}
The subsets $Y$ and $W$ from Example~\ref{ExampleLinePartition} give a \emph{multi-partition} of $\R^2$ in the sense of \cite{SchickEsfahani}.
The isomorphism from \eqref{IsoRoe} is then a simple example for the isomorphism of Thm.~1.4 ibid.
\end{remark}

A version of the following lemma can be found in the reference \cite[Lemma~6.7]{LudewigThiangCobordism}.

\begin{lemma}
\label{LemmaCalderon}
Let $u \in C^*_{\Roe}(\partial Y)^+ \otimes M_n(\C)$ be a unitary.
Then we have the ``Kubo formula''
\begin{equation}
\label{CalderonFormula}
  \theta_W([u]) = \mathrm{Tr}\bigl(u[\chi_W, u^*]\bigr),
\end{equation}
provided that the operator on the right hand side is trace-class.
\end{lemma}

\begin{proof}

With a view on \eqref{FFinverse}, the trace-class assumption on $u[\chi_W, u^*]$ implies that $F_u F_u^* - 1$ is trace-class.
By Calderon's formula (see, e.g., \cite[Lemma~4.1]{RoeOpen}), we therefore have 
\begin{equation*}
\begin{aligned}
&\mathrm{ind}(F_u) = \Tr(F_uF_{u^*}  - F_{u^*}F_u)\\
&= \Tr(\chi_W u \chi_W u^* \chi_W - \chi_W u^* \chi_W u \chi_W) \\
&=  -\Tr\bigl(\chi_W u (1-\chi_W) u^* - \chi_W u^* (1-\chi_W) u\bigr),
\end{aligned}
\end{equation*}
using cyclicity of the trace.
To continue, observe that
\begin{equation*}
  [\chi_W, u^*] = \chi_Wu^*(1-\chi_W) - (1-\chi_W)u^*\chi_W.
\end{equation*}
Because the composition of the two terms on the right hand side in any order is zero, they must be individually trace-class, as $u[\chi_W, u]$ (and hence also $[\chi_W, u^*]$) is trace-class.
We may therefore pull apart the trace above, to obtain
\begin{equation*}
\begin{aligned}
& -\Tr\bigl(\chi_W u (1-\chi_W) u^*\bigr) + \underbrace{\Tr\bigl(\chi_W u^* (1-\chi_W) u\bigr)}_{=\Tr((1-\chi_W) u\chi_W u^*)}\\
&\qquad= -\Tr\bigl(\chi_W u (1-\chi_W) u^* - (1-\chi_W) u\chi_W u^* \bigr)\\
&\qquad= \Tr(u[\chi_W, u^*]).
\end{aligned}
\end{equation*}
Here the underbraced identity follows again from cyclicity of the trace.
\end{proof}

We now give a physical interpretation of the functional $\theta_W$.
We treat the case that the Hamiltonian is bounded.
Hence let $H \in C^*_{\Roe}(X)$ be an insulator at energy level $E$ (Definition~\ref{DefinitionInsulator}) and let $p_E$ be the corresponding Fermi projection.
Let $\tilde{H} \in C^*_{\Roe}(\{Y\})$ be any Hamiltonian on $Y$ adapted to $H$ (Definition~\ref{DefinitionAdaptedOperator}).

\begin{theorem}
\label{ThmCurrentChannel}
Let $H$ be insulating at energy level $E$ with Fermi projection $p_E$.
Let $(a, b) \ni E$ be an open interval disjoint from the spectrum of $H$ and let $\varphi$ be a continuous function with integral one and compact support in $(a, b)$.
Then
\begin{equation*}
  \theta_W\bigl(\partial([p_E])\bigr) = - 2\pi \cdot \mathrm{Tr} \bigl(\varphi(\tilde{H}) \cdot i[\tilde{H}, \chi_W]\bigr),
\end{equation*}
provided that the trace on the right hand side exists.
Here $\partial$ is the Mayer-Vietoris boundary map \eqref{MayerVietorisBoundaryMap}. 
\end{theorem}

Similar formula can be found in \cite{ElbauGraf,KellendonkRichterSchulzBaldes}; see also \cite[Prop.~7.1.2]{ProdanSchulzBaldes}.
\S6 of \cite{LudewigThiangCobordism} treats the case where $H$ is a magnetic Schr\"odinger operator \eqref{MagneticSchroedinger}, where one obtains a similar formula.
Below we give a rather simple proof of Thm.~\ref{ThmCurrentChannel}.

\begin{interpretation}
When considering the boundary states of $\tilde{H}$ with energies lying in $(a, b)$, the term $i[\tilde{H}, \chi_W]$ is the time-derivative of the observable $\chi_W$ of being in $W$, by Heisenberg's equation of motion. 
We interpret $\varphi(\tilde{H})$ as a statistical ensemble of generalized eigenstates of $\tilde{H}$ with energies within $(a, b)$.
So $\Tr(\varphi(\tilde{H}) \cdot i[\tilde{H}, \chi_W])$ is the expected rate of change of probability to be inside $W$, within the statistical ensemble $\varphi(\tilde{H})$ of boundary localized states. 
Because this expectation value is $(2\pi)^{-1}$ times some Fredholm index, we deduce, a posteriori, that
the $(a, b)$-filling boundary states of $\tilde{H}$ constitute a quantized current channel flowing between $W$ and its complement.
\end{interpretation}

\begin{proof}[Proof of Thm.~\ref{ThmCurrentChannel}]
Let $\lambda_0$ be a lower bound on the spectrum of both $H$ and $\tilde{H}$ and let $f \in C^1_c(\R)$ such that $f^\prime(\lambda) = - \varphi(\lambda)$ for all $\lambda \geq \lambda_0$.
Then by choice of $\varphi$, we have $p_E = f(H)$. 
Recall from the proof of Thm.~\ref{ThmGapFilling} that $\partial([p_E]) = \Exp(r([p_E]))$, where $r$ is the restriction homomorphism \eqref{Homomorphismr} and $\Exp$ is the boundary map for the short exact sequence \eqref{SESpartialY}.
By the explicit formula \eqref{FormulaForExponentialMap} for $\Exp$ and the fact that $\tilde{H}$ is adapted, we therefore have 
\begin{equation*}
\partial([p_E]) = \partial([f(H)]) = \Exp(r([f(H)])) = [u],
\end{equation*}
 where
\begin{equation}
\label{Definitionu}
   u = \exp\bigl(2 \pi i \, f(\tilde{H})\bigr).
\end{equation}
With a view on Lemma~\ref{LemmaCalderon}, we are now done if we can show that $u[\chi_W, \tilde{u}]$ is trace-class if and only if $2\pi \cdot \varphi(\tilde{H})\cdot i [\chi_W, \tilde{H}]$ is, and that they have the same trace if they are trace-class.
If $[\chi_W, \tilde{H}]$ would commute with $\tilde{H}$, then this would just follow from the chain rule (as taking the commutator with $\chi_W$ is a derivation on the algebra of bounded operators on $\mathscr{H}$).
To establish this in general, we have to be careful.

Suppose that $u[\chi_W, u^*]$ is trace-class, so that \eqref{CalderonFormula} holds.
In a first step, with a view on \eqref{Definitionu}, we use the formula
\begin{equation}
\label{DerivativeOfExp}
  [\chi_W, e^a] = \int_0^1 e^{(1-t)a} [\chi_W, a] e^{ta} dt
\end{equation}
with $a = -2 \pi i\, f(\tilde{H})$, as well as cyclicity of the trace to obtain that 
\begin{equation*}
  \Tr\bigl(u[\chi_W, u^*]\bigr) = - 2 \pi i \cdot \Tr\bigl([\chi_W, f(\tilde{H})]\bigr).
\end{equation*}
Next, we use the identity
\begin{equation}
\label{FouriertransformIdentity}
f(\tilde{H}) = \frac{1}{2\pi} \int_{-\infty}^\infty \widehat{f}(s) e^{i s \tilde{H}} ds,
\end{equation}
valid for any Schwartz function of $\tilde{H}$, where $\hat{f}$ is the Fourier transform of $f$.
 By  \eqref{DerivativeOfExp} again, this time with $a=is \tilde{H}$, we get
\begin{equation*}
\begin{aligned}
 &[\chi_W, f(\tilde{H})] = \frac{1}{2\pi} \int_{-\infty}^\infty \hat{f}(s) [\chi_W, e^{i s \tilde{H}}] ds\\
 &~~= \frac{1}{2\pi} \int_{-\infty}^\infty\int_0^1 \widehat{f^\prime}(s)  e^{is(1-t) \tilde{H}}[\chi_W, \tilde{H}] e^{ist\tilde{H}} dt ds,
\end{aligned}
\end{equation*}
where we employed the identity $is\widehat{f}(s) = \widehat{f^\prime}(s)$.
Taking the trace and using its cyclicity, the $t$-integral collapses and we can use \eqref{FouriertransformIdentity} backwards (for $f^\prime$ instead of $f$) to obtain the desired identity, keeping in mind that $f^\prime = -\varphi$ on the spectrum of $\tilde{H}$.

Conversely, tracing back the previous calculations, we see that the trace-class property of $\varphi(\tilde{H})[\tilde{H}, \chi_W]$ implies that of $u[\chi_W, u^*]$.
\end{proof}

Observe that since $\varphi$ is supported in the spectral gap of $H$, we have $\varphi(H) = 0$. 
The fact that $\tilde{H}$ is adapted to $H$ therefore implies that $\varphi(\tilde{H}) \in C^*_{\Roe}(\partial Y)$.
On the other hand, by Lemma~\ref{LemmaCommutator}, $[\tilde{H}, \chi_W] \in C^*_{\Roe}(\partial W)$.
So $\varphi(\tilde{H}) \cdot i[\tilde{H}, \chi_W]$ is supported on $\partial Y \cap \partial W$ which is bounded by the transversality assumption on $Y$ and $W$; hence $\varphi(\tilde{H}) \cdot i[\tilde{H}, \chi_W]$ is a compact operator.

We now give a sufficient criterion for this operator to be even trace-class, so that Thm.~\ref{ThmCurrentChannel} applies.
Suppose that $X = \Z^d$ with the $X$-module $\mathscr{H} = \ell^2(\Z^d) \otimes \C^n$ for some coefficient Hilbert space $\mathscr{K}$.
Let $H = (H_{ij})_{i, j \in \Z^d}$, be a Hamiltonian with rapidly decaying coefficients, meaning that for each $\mu \geq 0$, one has
\begin{equation}
\label{RapidDecay}
 \sup_{i, j \in \Z^d} \|H_{ij}\| \cdot d(i, j)^\mu < \infty.
\end{equation} 
for all $i, j \in \Z^d$.
Let $\tilde{H} = \chi_Y H \chi_Y$ be the adapted Hamiltonian on $Y$ from Example~\ref{ExampleContraction}.
We assume that $Y$ and $W$ are the standard half spaces \eqref{StandardHalfSpaces}.

\begin{theorem}
Under the above assumptions, the operator $\varphi(\tilde{H}) \cdot i[\tilde{H}, \chi_W]$ from Thm.~\ref{ThmCurrentChannel} is trace-class.
\end{theorem}

\begin{proof}
Let $\mathscr{C}_{\mathrm{rd}}(\Z^d) \subset C^*_{\Roe}(\Z^d)$ be the \emph{rapid decay algebra} consisting of operators whose matrix coefficients satisfy the decay property \eqref{RapidDecay}.
As observed in \cite[\S2.16]{EwertMeyer}, $\mathscr{C}_{\mathrm{rd}}(\Z^d)$ is closed under smooth functional calculus for normal elements.
In particular, $\varphi(\tilde{H}) \in \mathscr{C}_{\mathrm{rd}}(\Z^d)$.
Using that $\varphi$ is supported in the spectral gap of $H$, one can moreoever show that $\varphi(\tilde{H})$ decays rapidly away from $\partial Y$, meaning that for each $Z \in \partial Y$ and each $\mu \geq 0$, one has
\begin{equation}
\label{DecayAwayFromZ}
 \sup_{i, j \in \Z^d} \|\varphi(\tilde{H})_{ij}\| \cdot d(i, Z)^\mu\cdot  d(j, Z)^\mu < \infty .
\end{equation}
Similarly, $[\chi_W, \tilde{H}]$ decays rapidly away from $\partial W$.
The product of these operators decays rapidly away from $\partial Y \cap \partial W$, which implies that the product is trace-class.
\end{proof}

Clearly, the above still holds if $Y$ and $W$ are modified in a bounded way, as this does not change the decay conditions \eqref{RapidDecay}, \eqref{DecayAwayFromZ}.

\vspace{1.8cm}


\bibliography{Literatur}

\end{document}